\documentclass[12pt, titlepage, reqno]{article}
\usepackage[super]{natbib}
\usepackage{amssymb, latexsym}
\usepackage{amsmath}
\usepackage{amsfonts}
\usepackage{graphicx, centernot}
\usepackage{amsthm}
\usepackage{setspace}
\usepackage{booktabs}
\usepackage{bm}
\usepackage{enumerate}
\usepackage[margin = 1in]{geometry}
\usepackage{titlesec}
\usepackage{fixfoot}
\usepackage{xspace}

\DeclareFixedFootnote{\eqc}{Equal contribution}
\titleformat*{\section}{\normalsize\bfseries}
\titleformat*{\subsection}{\normalsize\bfseries}

\theoremstyle{plain}
\newcommand{\bL}{\textbf{\textit{L}}}

\newcommand{\E}{\text{E}}

\newcommand{\ci}{\mathrel{\text{\scalebox{1.07}{$\perp\mkern-10mu\perp$}}}}
\newcommand{\Pro}{\text{P}}

\newcommand{\btheta}{\boldsymbol{\theta}}

\newcommand{\dI}{\Delta_{\text{ITT}}}
\newcommand{\dIhat}{\widehat{\Delta}_{\text{ITT}}}
\newcommand{\V}{\text{Var}}

\newtheorem{theorem}{Theorem}
\newtheorem{lemma}{Lemma}
\newtheorem{corollary}{Corollary}

\begin{document}

\vspace*{0.05in}

\thispagestyle{empty}

\begin{center}

\begin{spacing}{2.0}
{\Large \textbf{Bounding the local average treatment effect in an instrumental variable analysis of engagement with a mobile intervention}} \\
\end{spacing}

\vspace{4ex}

Andrew J. Spieker, Robert A. Greevy, Lyndsay A. Nelson, and Lindsay S. Mayberry\\
 
\vspace{5mm}

Vanderbilt University Medical Center \\

Nashville, TN 37203 \\

\end{center}

\begin{spacing}{1.0}
\noindent \textbf{Abstract}: Estimation of local average treatment effects in randomized trials typically requires an assumption known as the exclusion restriction in cases where we are unwilling to rule out unmeasured confounding. Under this assumption, any benefit from treatment would be mediated through the post-randomization variable being conditioned upon, and would be directly attributable to neither the randomization itself nor its latent descendants. Recently, there has been interest in mobile health interventions to provide healthcare support; such studies can feature one-way content and/or two-way content, the latter of which allowing subjects to engage with the intervention in a way that can be objectively measured on a subject-specific level (e.g., proportion of text messages receiving a response). It is hence highly likely that a benefit achieved by the intervention could be explained in part by receipt of the intervention content and in part by engaging with/responding to it. When seeking to characterize average causal effects conditional on post-randomization engagement, the exclusion restriction is therefore all but surely violated. In this paper, we propose a conceptually intuitive sensitivity analysis procedure for this setting that gives rise to sharp bounds on local average treatment effects. A wide array of simulation studies reveal this approach to have very good finite-sample behavior and to recover local average treatment effects under correct specification of the sensitivity parameter. We apply our methodology to a randomized trial evaluating a text message-delivered intervention for Type 2 diabetes self-care.
\end{spacing}

\clearpage

\addtocounter{page}{-1}

\begin{spacing}{2.0}

\noindent \section{Introduction}

There has been recent interest in studies of mobile (e.g., text message-based) interventions designed to improve health outcomes (e.g., hemoglobin A1c) by targeting self-efficacy and self-care behaviors such as medication adherence.\cite{Greenwood17, Marcolino18}  Rapid Encouragement/Education and Communications for Health (REACH), for instance, is a text message-delivered intervention designed to support medication adherence for patients with Type 2 diabetes.\cite{Nelson18} The REACH study sought to evaluate the effects of this intervention on hemoglobin A1c (HbA1c) as compared to a control condition. A key feature of the REACH study is that subjects in the intervention arm received both one-way text messages providing information and interactive (two-way) text messages requesting a response. Though subjects in the intervention arm received the same number of messages, there was variation in response rate across study subjects. The extent to which a subject responds to interactive messages serves as an objective measure of his or her engagement with the intervention. Natural goals therefore include determining the effect of REACH conditional on engagement and, relatedly, characterizing the extent to which the effect of REACH is determined by engagement with the intervention.

The post-randomization nature of engagement, together with unmeasured common causes of engagement and HbA1c, obscure our ability to achieve these goals with standard regression techniques. Instrumental variable (IV) methods, commonly used in randomized trials to evaluate causal effects in settings of noncompliance, are designed to address this challenge.\cite{Imbens94, Angrist95, Angrist96, Frangakis02, Roy08} Though conceptually distinct, compliance and engagement are structurally analogous from an analytical standpoint. Throughout this paper, we will therefore use engagement with REACH as an anchor for describing the IV framework.

In the simple case where subject-specific engagement is considered dichotomously (e.g., response to at least 80\% of text messages), the target parameter of a traditional IV method involves a comparison of mean potential outcomes among individuals for whom engagement and treatment assignment are entirely concordant (that is, among those who would not engage when randomized to the control arm, but would engage when assigned to the REACH intervention). It is for this reason that an IV method is said to estimate a \textit{local average treatment effect}, setting it apart from other commonly used causal inference approaches such as standardization and weighting.\cite{Rosenbaum1983, Robins86, Robins00, Lunceford04, Funk11} The typical framework for an IV can be simplified in our setting in the sense that those not receiving the REACH intervention are unable to engage with it. Put another way, the subjects in the control arm have a known engagement level of identically zero, ensuring that the assumption of treatment monotonicity (that is, that engagement under the intervention is at least as high as engagement under the control) is satisfied. Monotonicity would otherwise be an untestable identifying assumption, necessary nevertheless.\cite{Frangakis02} This simplification allows us to characterize local average treatment effects across levels of a continuous measure of engagement (e.g., proportion of messages receiving a response), as the ``locality" is uniquely indexed by a single variable.

Despite this simplification, application of a traditional IV method in this setting faces the barrier of almost certain violations to an assumption known as the \textit{exlusion restriction}. The exclusion restriction states that any benefit from treatment would be mediated through engagement, and would not be derived directly through randomization itself or mediated through any of its latent descendants. Since the receipt of content may motivate or cue self-care behavior irrespective of a subject's choice to respond to it, this assumption is tenuous at best in the case of REACH or other similar interventions.

Although local average treatment effects are not generally identifiable under violations to the exclusion restriction when there is unmeasured confounding, the focus of this paper will be to outline a sensitivity analysis procedure based on an intuitive, conceptually straightforward parameter. We will further derive and justify conditions under which local average treatment effects can be bounded, and demonstrate how our sensitivity approach extends to evaluation of treatment effect heterogeneity by engagement. The remainder of this paper is organized as follows. In Section 2, we provide a description of our notation and assumptions, and define the local average treatment effects of interest in terms of an intuitive sensitivity parameter. In Section 3, we characterize the resulting bounds on such effects. In Section 4, we propose estimation and inferential procedures, as well as a framework for evaluating treatment effect heterogeneity. In Section 5, we conduct a simulation study in order to evaluate the finite-sample performance of our proposed procedure, and in Section 6, we apply our results to the REACH study. We conclude in Section 7 with a discussion of our findings and possible future directions for methodological research.

\noindent \section{Definitions, assumptions, and weak identifiability}

In this section, we provide an outline of our notation, define a class of local average treatment effects of interest, and characterize the assumptions necessary to bound such effects.

\subsection{Notation}

Let $i = 1, \dots, N$ denote the independently sampled study subjects. We let $Z$ denote binary randomization (which serves as the instrument), $A$ the observed engagement variable, and $Y$ the observed outcome. Without loss of generality, we assume higher values of $A$ to signify higher levels of engagement (with $A = 0$ signifying no engagement). Following the potential outcomes framework of Rubin, let $A^{z = 0}$ and $A^{z = 1}$ denote the potential engagement status under randomization to treatment $z = 0$ and $z = 1$, respectively; similarly, let $Y^{z = 0}$ and $Y^{z = 1}$ denote the potential outcomes under each respective treatment.\cite{Rubin74} Figure 1 depicts a directed acyclic graph (DAG) illustrating the temporal ordering of these observed variables. The ITT effect is defined as $\dI = \E[Y^{z = 1}] - \E[Y^{z = 0}]$; owing to randomization of $Z$, this quantity can be identified and is readily expressed as $\dI = \E[Y|Z = 1] - \E[Y|Z = 0]$ under assumptions presented in Section 2.3 (to be discussed).

\begin{figure}[h!]
\centering
\includegraphics[width = 2.8in]{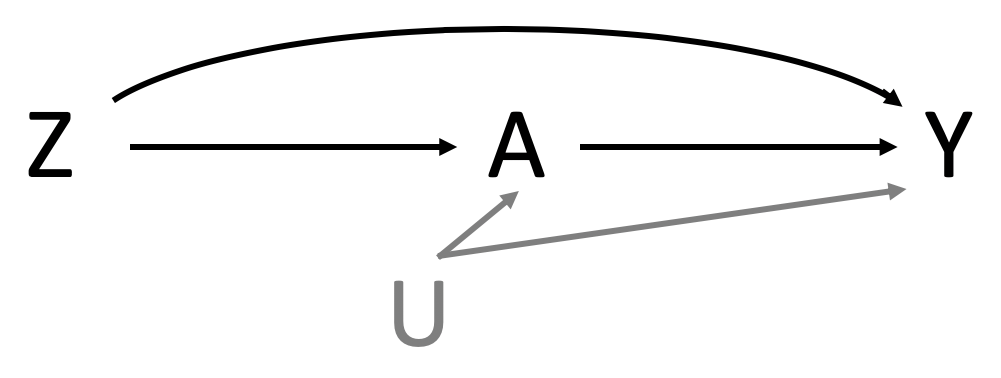}
\caption{Directed acyclic graph depicting the temporal ordering and causal pathways between measured variables. Note that $U$ (gray), is a collection of unobserved confounders impacting engagement and the outcome. Note that the direct path from $Z$ to $Y$ cannot reasonably be ruled out in the setting of a mobile health intervention, and serves as a violation of the exclusion restriction.}
\end{figure}

\subsection{Engagement-compliance and local average treatment effects}

First, consider $A$ to be binary for simplicity (we will later not require this). We may use combinations of $A^{z = 0}$ and $A^{z = 1}$ (i.e., principal stratification) to partition the population into four hypothetical engagement classes (Table 1).\cite{Angrist96, Frangakis02} For instance, a subject who would not engage under randomization to the control ($Z = 0$), but would engage under randomization to the intervention ($Z = 1$), could be referred to as \textit{engagement-compliant}. The monotonicity assumption is equivalent to assuming no engagement-defiant subjects; in our example, this assumption is, in essence, trivially guaranteed as it is impossible to engage with an intervention that is not received (that is, $Z = 0 \Rightarrow A = 0$). By similar logic, this specific setting under discussion allows us to further preclude the presence of \textit{always-engagers}. We may therefore partition the population into principal strata on the basis of $A^{z = 1}$ alone, terming those in the population with $A^{z = 1} = 1$ as \textit{engagement-compliant}, and those for whom $A^{z = 1} = 0$ as \textit{never-engagers}. Under our framework, compliance class is latent in the control group, and observed in the intervention group, as $A^{z = 1}$ is observed in subjects receiving the intervention and uniquely characterizes compliance class (as compared to the usual setting of compliance, in which latency of $A^{z}$ is not specific to $z$).

\begin{table}[h!]
\caption{Characterization of engagement-compliance classes (principal strata) in the simple setting considering engagement as a binary variable. Note that only two such classes are applicable in our particular setting, uniquely defined by $A^{z = 1}$.}
$$\begin{tabular}{cclc}
	$A^{z = 0}$ & $A^{z = 1}$ & Characterization & Applicable? \\ \hline
	0 & 0 & Never-engager & Yes \\
	0 & 1 & Engagement-compliant & Yes \\ 
	1 & 0 & Engagement-defiant & No \\ 
	1 & 1 & Always-engager & No \\  \hline
\end{tabular}$$
\end{table}

Now, suppose that $A^{z = 1}$ is continuous, with $0 \leq A^{z = 1} \leq 1$. We define the following class of local average treatment effects, uniquely indexed by $A^{z = 1}$:
\begin{eqnarray*}
\Delta(a) &=& \E[Y^{z = 1} - Y^{z = 0}|A^{z = 1} = a].
\end{eqnarray*}
In plain language, $\Delta(a)$ denotes the average causal effect of randomization on the outcome of interest among a subpopulation having some specified hypothetical level of engagement, $a$, under treatment $Z = 1$. We refer specifically to $\Delta(0)$ as the never-engager causal effect (NECE) and $\Delta(1)$ as the engagement-compliant causal effect (ECCE).

\subsection{Assumptions}

We invoke the following assumptions in order to formulate bounds on the local average treatment effects characterized in Section 2.2.
\begin{enumerate}
\item Stable unit treatment value assumption (SUTVA): $(A_i^{z}, Y_i^{z}) \ci Z_j$.
\item Consistency: $(A, Y) = (A^Z, Y^Z)$.
\item Positivity: $0 < \Pro(Z = 1) < 1$.
\item Ignorability of randomization: $Y^{z} \ci Z$ and $A^z \ci Z$ for each $z = 0, 1$.
\item Monotonicity: $\E[A^{z = 1}] \geq \E[A^{z = 0}]$.
\item Instrument validity: $Z \not\ci A$.
\item The $\gamma$-principle: $0 \leq a < a' \leq 1 \Rightarrow \text{sgn}(\Delta(a)) = \text{sgn}(\Delta(a'))$, and $|\Delta(a)| \leq |\Delta(a')|$.
\end{enumerate}

Each of these assumptions can be described as follows: (1) SUTVA states that the treatment assignment of one individual does not influence the potential engagement or outcome of another individual; a violation to SUTVA is often referred to as \textit{interference}; vaccine trials serve as a particularly well known area in which interference commonly poses challenges.\cite{Hudgens08} (2) Consistency implies that the observed level of engagement and the observed outcome correspond to the potential engagement and outcome under the randomization actually received. (3) Positivity refers to a nonzero probability of assignment to each treatment group. (4) Ignorability of randomization, also known as exchangeability, implies unconfoundedness of the relationship between randomized treatment and both the engagement and outcome measures. Note that the possibility of unmeasured confounders for the relationship between \textit{engagement} and the outcome is not precluded (Figure 1). Consistency, positivity, and ignorability can all be reasonably assumed under proper randomization and outcome measurement. Assumptions 1-4 together ensure identifiability of $\dI$.

Recall that in our specific setting of engagement with a mobile health intervention, we have that $Z = 0 \Rightarrow A = 0$. In such cases, monotonicity (Assumption 5) is implied by consistency, and instrument validity (Assumption 6) can be expressed as $\mu_A = \E[A^{z = 1}] = \Pro(A^{z = 1} = 1) > 0$, with higher values of $\mu_A$ signifying higher instrument strength. We have substituted the ordinary exclusion restriction (namely, that $Z \ci Y|A$) with Assumption 7, which we term \textit{the $\gamma$-principle}. It is so named as it implies boundedness of the value of $\gamma = \Delta(0)/\Delta(1)$ between zero and one, although its specific value cannot be identified (by convention, assume that $\Delta(1) = 0 \Rightarrow \gamma = 0$). Hence, $\gamma$ will serve as a sensitivity parameter to be varied over its possible range of values. The $\gamma$-principle can be realized as a generalization of the exclusion restriction, under which $\gamma = \Delta(0) = 0$, and allows us to bound $\Delta(a)$ for $0 < a < 1$ (these bounds will be the subject of Section 3).

\subsection{Weak identifiability of $\boldsymbol{\Delta(a)}$}

First, express $\Delta(a) = \Delta(1)g(a)$ for $0 \leq a \leq 1$. The $\gamma$-principle is equivalent to expressing $g(a) = \gamma + (1 - \gamma)h(a)$ for some $\gamma \in [0, 1]$ and some monotone increasing $h(\cdot)$ with $h(0) = 0$ and $h(1) = 1$. Setting $a = 0$ shows that $\gamma = \Delta(0)/\Delta(1)$ possesses the conceptually intuitive interpretation as the ratio of the NECE to the ECCE. This parameterization relaxes the ``through-the-origin" relationship presumed under the exclusion restriction, under which $g(0) = 0$. We will demonstrate that $\Delta(a)$ is identifiable under correct specification of the non-identifiable $\gamma$ and $h(\cdot)$; we refer to this as \textit{weak identifiability} of $\Delta(a)$.

Expressing $\Delta(a) = \Delta(1)[\gamma + (1 - \gamma)h(a)]$, $\Delta(1)$ may be expressed in terms of $\dI$, $\gamma$, and $\mu_{h} = \E[h(A^{z = 1})] = \E[h(A)|Z = 1]$ as follows:
\begin{eqnarray*}
	\dI = \E[Y|Z = 1] - \E[Y|Z = 0] &=& \E[Y^{z = 1} - Y^{z = 0}]\\
	~ &=& \E_{A^{z = 1}}[\E[Y^{z = 1} - Y^{z = 0}|A^{z = 1}]]\\
	~ &=& \E_{A^{z = 1}}[\Delta(A^{z = 1})]\\
	~ &=& \E_{A^{z = 1}}\left[\Delta(1)\left\lbrace\gamma + (1 - \gamma)h(A^{z = 1})\right\rbrace\right]\\
	~ &=& \Delta(1)\left\lbrace\gamma + (1 - \gamma)\mu_h\right\rbrace.
\end{eqnarray*}
Rearranging,
\begin{eqnarray}
	\Delta(1) &=& \frac{\E[Y|Z = 1] - \E[Y|Z = 0]}{\gamma + (1 - \gamma)\mu_{h}} = \frac{\dI}{\gamma + (1 - \gamma)\mu_{h}}.
\end{eqnarray}
Since $\dI$ is itself identifiable under Assumptions 1-4 in Section 2.3, correct specification of $h(\cdot)$ and $\gamma$, implies the identifiability of $\Delta(1)$. In turn, the subsequent weak identifiability of local average treatment effects for general levels of engagement follows directly from the parameterization for $\Delta(a)$ posed above:
\begin{eqnarray}
	\Delta(a) = \dI\times \frac{\gamma + (1 - \gamma)h(a)}{\gamma + (1 - \gamma)\mu_{h}} = \dI \times c_{\gamma; h}(a),
\end{eqnarray}

It follows that $\Delta(a)$ possesses a stationary property: $\Delta(\mu_{h}) = \dI$ for all $\gamma$. Further, weak identifiability of $\Delta(a)$ only requires Assumptions 1-6 to hold. Bounds on $\gamma$ implied by Assumption 7 will serve to place sharp bounds $\Delta(a)$, discussed in Section 3.

\subsection{Considerations regarding parameterizations of $\boldsymbol{\Delta(a)}$}

While the exclusion restriction has been well described in many applications of traditional IV approaches, correct specification of $h(\cdot)$ is a key assumption that merits discussion. Commonly, this assumption is implicitly expressed through dichotomization of a continuous $A$ in order to characterize target parameters based on discrete principal strata. In the framework of Section 2.4, this is achieved by defining $h(a) = 1(a > \zeta)$ for some $\zeta$, implying minimal treatment benefit to be derived among all for whom $A^{z = 1} \leq \zeta$, and maximal treatment benefit to be derived among all for whom $A^{z = 1} > \zeta$.

Angrist and Imbens show that discretization of an $A$ having a continuous dose-response relationship tends to bias estimates of $\Delta(1)$ away from the null.\cite{Angrist95} At the same time, identification of $h(\cdot)$ under continuous $A$ is not possible absent a continuous instrument, $Z$. In such settings, they propose characterizing local average treatment effects linearly across values of $A$ (in their work, via two-stage least-squares). Following suit, we will therefore focus our attention on cases in which $h(\cdot)$ is chosen to be the identity function throughout our simulation study and application. However, we will develop theory for general choices of $h(\cdot)$ suitable for circumstances that suggest non-identity $h(\cdot)$ to be more appropriate.  

\noindent \section{Bounding local average treatment effects}

As previously discussed, the exclusion restriction invoked by a classic IV analysis forces the condition $\Delta(0) = 0$; the $\gamma$-principle relaxes this assumption, assuming instead that $\Delta(0) = \gamma \times \Delta(1)$ for some $\gamma \in [0, 1]$. It follows directly that $\Delta(a)$ achieves global extrema at $\gamma = 0$ and $\gamma = 1$ for fixed values of $a$; these extrema are further unique if and only if $\dI \neq 0$. Assuming without loss of generality that $\dI$ is non-negative, the bounds for $\Delta(a)$ can be characterized as follows:
\begin{eqnarray*}
	    \dI \times \frac{h(a)}{\mu_{h}} \leq & \Delta(a) & \leq \dI \hspace{0.3in} \mbox{ for } a \leq \mu_{h} \\
	    & \Delta(a) & =  \dI \hspace{0.3in} \mbox{ for } a = \mu_h \\
		\dI \leq & \Delta(a) & \leq \dI \times \frac{h(a)}{\mu_{h}}  \hspace{0.3in} \mbox{ for } a > \mu_{h}.
\end{eqnarray*}

In the specific case of $a = 0$, the NECE is bounded by zero (below) and $\dI$ (above); when $a = 1$, the ECCE is bounded by $\dI$ (below) and the Wald formula (above):
\begin{eqnarray*}
	0 \leq \Delta(0) \leq \dI \leq \Delta(1) \leq \frac{\dI}{\mu_{h}}.
\end{eqnarray*}
If $\dI$ is non-positive, the directionality of each of these inequalities is reversed. These treatment effect bounds are illustrated in Figure 2 in the case where $h(a) = a$.

\begin{figure}[h!]
	\centering
	\includegraphics[width = 6.5in]{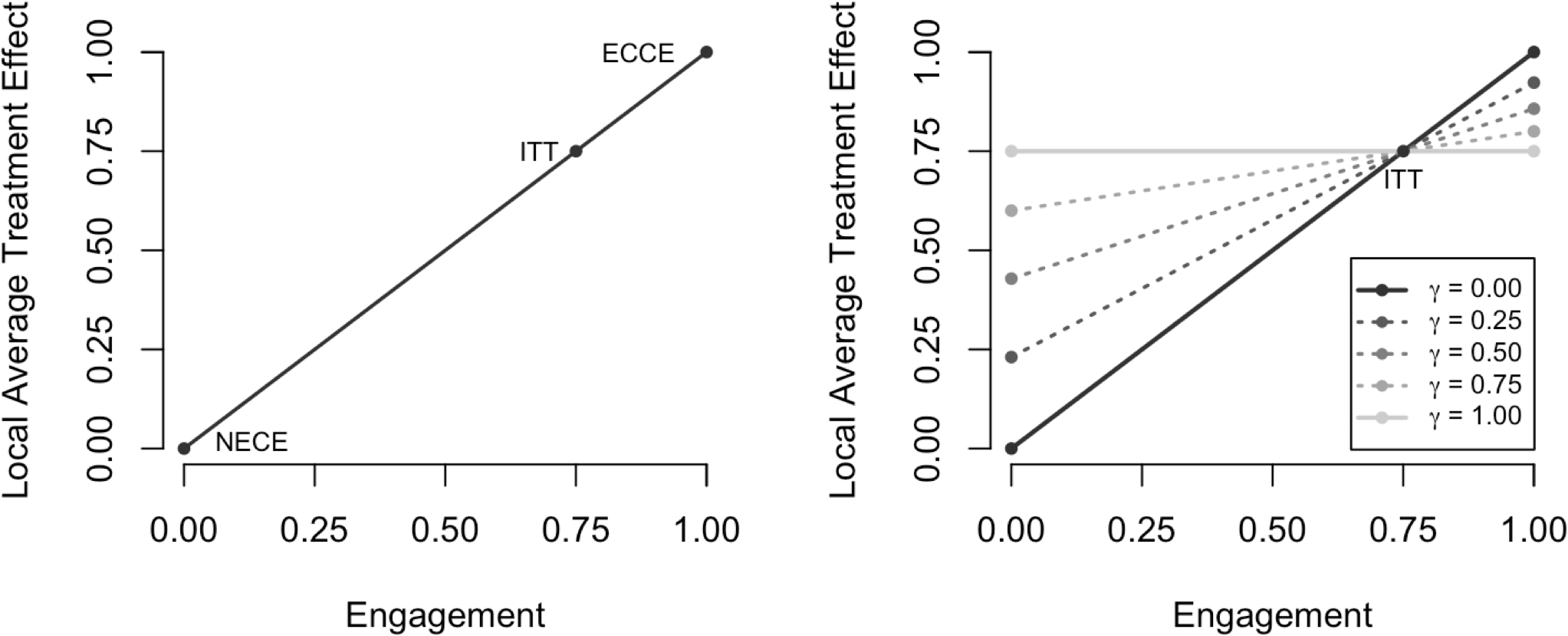}
	\caption{Illustration of treatment effect bounds in a simple setting. On the left, we highlight how the treatment effect would be characterized across different levels of engagement by a traditional IV analysis ($\gamma = 0$). The presumed NECE, ECCE, and ITT are all shown in this case. On the right, we show how the treatment effect is characterized across different levels of engagement under various sensitivity parameters. Note at each level of engagement, the treatment is bounded by the ITT and a specified multiple thereof (solid lines); also depicted are the results for nontrivial values of $\gamma$ (dotted lines).}
\end{figure}

Also of interest is to understand the behavior of local average treatment effects across $\gamma$ for different instrument strengths (as defined by $\mu_h$ in this case). This behavior is illustrated in Figure 3, again considering the case where $h(\cdot)$ is the identity function. Linearity of $\Delta(a)$ in $h(a)$ for fixed values of $\gamma$ does not imply general linearity of $\Delta(a)$ in $\gamma$ for fixed values of $h(a)$. In fact, the latter condition only holds in the presence of a perfect instrument. This suggests that under a weak instrument (low engagement), a sensitivity analysis of $\Delta(0)$ and $\Delta(1)$ can be expected to produce greater fluctuations when varying values of $\gamma$ closer to zero as compared to values closer to one; under a stronger instrument (higher engagement), sensitivity will be closer to constant across $\gamma$.

\begin{figure}[h!]
	\centering
	\includegraphics[width = 6.5in]{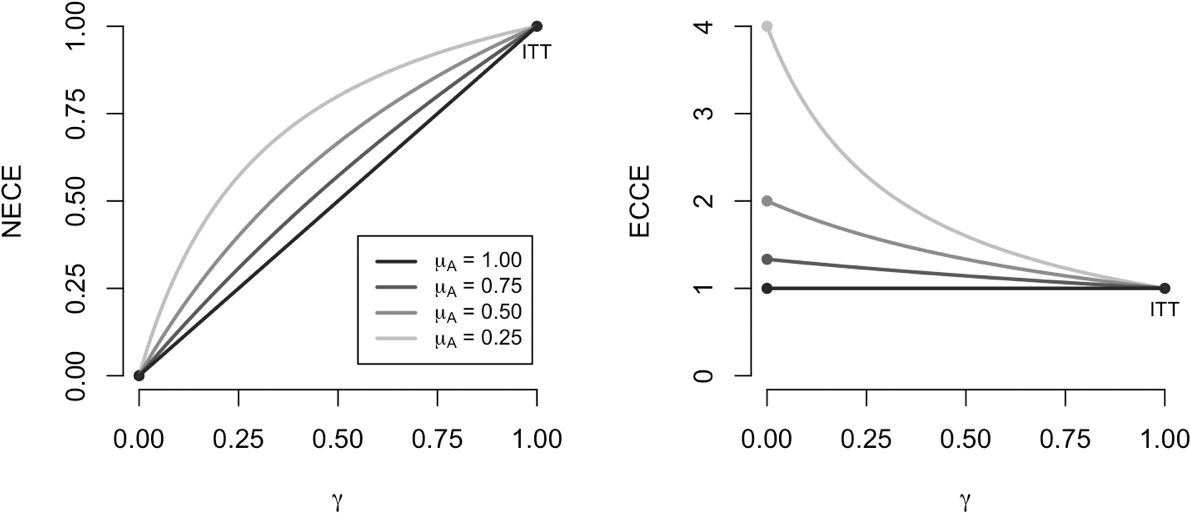}
	\caption{Illustration of treatment effect bounds in a simple setting. On the left, we highlight how the NECE, $\Delta(0)$, varies across the levels of $\gamma$ for different values of $\mu_A$, and on the right we similarly depict the ECCE, $\Delta(1)$. Note that $\mu_{A}$ characterizes instrument strength; weaker instruments are associated with higher curvature.}
\end{figure}

\noindent \section{Estimation and inference}

The problem of estimating $\Delta(a)$ can be decomposed into the following steps: (1) specification of a value for $\gamma$, (2) specification of a form for $h(\cdot)$, (3) estimation of $\dI$ and $\mu_{h}$, and (4) plugging in estimates from the previous step into Equation (2) of Section 2.4. We must distinguish between the form of $h(\cdot)$ and value of $\gamma$ that correspond to the unknowable data generating mechanism, and the values that are specified by the user. We will let $h_0(\cdot)$ and $\gamma_0$ correspond to the true underlying mechanism, and use the notation $\widehat{\Delta}_{h; \gamma}(a)$ to denote an estimator of $\Delta(a)$ under the user-specified sensitivity parameter, $\gamma$, and transformation, $h$. We will let $\Delta_{h; \gamma}(a)$ denote the value for which $\widehat{\Delta}_{h; \gamma}(a)$ is consistent---which may or may not be equal to $\Delta(a) = \Delta_{h_0, \gamma_0}(a)$, depending on correctness of choices for $h(\cdot)$ and $\gamma$. A simple estimator utilizes the corresponding sample means in the obvious way:
\begin{eqnarray*}
\widehat{\Delta}_{h; \gamma}(a) &=& \dIhat \times \widehat{c}_{\gamma; h}(a) = \left(\overline{Y}_{Z = 1} - \overline{Y}_{Z = 0}\right) \times \frac{\gamma + (1 - \gamma)h(a)}{\gamma + (1 - \gamma)\overline{h(A)}_{Z = 1}}.
\end{eqnarray*}
Owing to maximal efficiency associated with $\overline{h(A)}_{Z = 1}$ as an estimator of $\mu_{h}$, there is no obvious incentive to consider alternative estimators. On the other hand, greater efficiency for estimation of $\dI$ intuitively corresponds to greater efficiency for estimation of $\Delta_{h; \gamma}(a)$. For instance, this could be achieved through adjustment for baseline covariates, $\bL$, in a linear regression model): $\E[Y|Z = z, \bL = \boldsymbol{\ell}] = \beta_0 + \beta_1 z + f_{\btheta}(\boldsymbol{\ell})$, where, $f_{\btheta}(\cdot)$ denotes a function of baseline covariates indexed by $\btheta$. Importantly, consistency of $\widehat{\beta}_1$ for $\dI$ does not depend upon correct specification of $f_{\btheta}$.\cite{Tsiatis08} Because $\dI$ can be estimated in multiple ways, we discuss asymptotic theory generally rather than under a specific estimator.

\begin{lemma}
Under the assumption of a valid instrument, $\widehat{c}_{\gamma; h}(a)$ achieves $\sqrt{N}$-consistency and asymptotic normality for $(\gamma, a) \neq (0, 0)$.
\end{lemma}

\begin{proof}
This is a standard and straightforward application of the Law of Large Numbers, the L\'{e}vy Central Limit Theorem, and $\delta$-method with $g_{\gamma; a}(\mu_h) = c_{\gamma; h}(a)$, and so we do not provide this proof in detail. Letting $\sigma_h^2$ denote $\V\left(\sqrt{N}(\widehat{\mu}_{h} - \mu_{h})\right)$, the asymptotic variance associated with $\sqrt{N}(\widehat{c}_{\gamma; h}(a) - c_{\gamma; h}(a))$ is given by:
\begin{eqnarray*}
\sigma_{c_{\gamma; h}}^2(a) = \frac{(1 - \gamma)^2(\gamma + (1 - \gamma)h(a))^2}{(\gamma + (1 - \gamma)\mu_{h})^4}\sigma_h^2.
\end{eqnarray*}
We do not require asymptotic theory under the condition that $\gamma = a = 0$, as $\widehat{\Delta}_0(0) = 0$ identically by convention.
\end{proof}

\begin{theorem}
Under both $\sqrt{N}$-consistency and asymptotic normality of $\dIhat$, we have that $\sqrt{N}(\widehat{\Delta}_{\gamma; h}(a) - \Delta_{\gamma; h}(a)) \longrightarrow_d \mathcal{N}(0, \tau_{\gamma; h}^2(a))$ for some $\tau_{\gamma; h}^2(a) > 0$.
\end{theorem}

\begin{proof}
We employ the following decomposition and invoke Slutsky's theorem:
\begin{eqnarray*}
	\sqrt{N}(\widehat{\Delta}_{\gamma; h}(a) - \Delta_{\gamma; h}(a)) &=& \sqrt{N}(\dIhat \times \widehat{c}_{\gamma; h}(a) - \dI \times c_{\gamma; h}(a)) \\
	~ &=&  c_{\gamma; h}(a) \left[\sqrt{N}(\dIhat - \dI)\right] + \dIhat \left[\sqrt{N}(\widehat{c}_{\gamma; h}(a) - c_{\gamma; h}(a)\right]\\
	~ &\approx& c_{\gamma; h}(a) \left[\sqrt{N}(\dIhat - \dI)\right] + \dI \left[\sqrt{N}(\widehat{c}_{\gamma; h}(a) - c_{\gamma; h}(a))\right].
\end{eqnarray*}
Letting $\sigma_{\text{ITT}}^2$ denote the asymptotic variance of $\sqrt{N}(\dIhat - \dI)$, noting that $\widehat{c}_{\gamma; h}(a)$ and $\dIhat$ are asymptotically uncorrelated, and invoking both Lemma 1 and Slutsky's theorem again, it follows that $\tau_{\gamma; h}^2( a) = [c_{\gamma; h}(a)]^2\sigma_{\text{ITT}}^2 + \Delta_{\text{ITT}}^2\sigma_{c_{\gamma; h}}^2(a)$.
\end{proof}

Note the following important corollaries; proofs of the first two are trivial and are hence not provided.

\begin{corollary}
If the user-specified sensitivity parameter, $\gamma$, and transformation $h(\cdot)$ are each ``chosen correctly" in the sense that $\gamma = \gamma_0 = \Delta(0)/\Delta(1)$ and $h(\cdot) = h_0(\cdot)$, then $\widehat{\Delta}_{\gamma; h}(a) \longrightarrow_p \Delta(a) \equiv \Delta_{\gamma_0; h_0}(a)$ for $0 \leq a \leq 1$.
\end{corollary}

\begin{corollary}
If $\dIhat$ and $\dIhat'$ denote two consistent estimators of $\dI$ such that $\sigma_{\text{ITT}}^2 < [\sigma_{\text{ITT}}']^2$, then the estimator $\widehat{\Delta}_{\gamma; h}(a)$ based on $\dIhat$ achieves greater asymptotic efficiency as compared to that based on $\dIhat'$.
\end{corollary}

\begin{corollary}
For all $(\gamma, a) \neq (0, 0)$, (1), $\tau_{\gamma; h}^2(a)$ is decreasing in $a$, and (2) a Wald test of the null hypothesis $H_0 : \Delta_{\gamma; h}(a) = 0$ is asymptotically equivalent to a Wald test of the null hypothesis $H_0 : \dI = 0$. That is, for sufficiently large $N$,
\begin{eqnarray*}
W_N^{\gamma; h} = N\left[\frac{\widehat{\Delta}_{\gamma; h}(a)}{\widehat{\tau}_{\gamma; h}(a)}\right]^2 \approx W_N^{\text{ITT}} = N\left[\frac{\dIhat}{\widehat{\sigma}_{\text{ITT}}}\right]^2
\end{eqnarray*}
\end{corollary}

\begin{proof}
	The first statement is trivial; the second follows by noting that $\Delta_{\gamma; h}(a) = 0 \Leftrightarrow \dI = 0$ when $(\gamma, a) \neq (0, 0)$; and then invoking Slutsky's theorem to note the asymptotic distribution of $\widehat{\Delta}_{\gamma; h}(a)$ under the null:
	\begin{eqnarray*}
		\sqrt{N}(\widehat{\Delta}_{\gamma; h}(a) - \Delta_{\gamma; h}(a)) &\overset{H_0}{\approx}& c_{\gamma; h}(a) \left[\sqrt{N}(\dIhat - \dI)\right],
	\end{eqnarray*}
	with asymptotic variance $\tau_{\gamma; h}^2(a) \approx c_{\gamma; h}^2(a)\sigma_{\text{ITT}}^2$. Therefore, Wald-based confidence intervals for both $\dI$ and $\Delta_{\gamma; h}(a)$ will possess the same coverage properties, asymptotically. Note that $W_N^{\gamma; h}, W_N^{\text{ITT}} \longrightarrow_d \chi_1^2$.
\end{proof}

While the underlying asymptotic theory plays an important role in developing confidence intervals and conducting hypothesis tests in large samples, we note that approximate normality cannot reasonably be assumed to hold in smaller samples, particularly due to ratio-based plug-in estimation associated with $\widehat{\Delta}_{\gamma; h}(a)$. We propose forming quantile-based confidence intervals and standard errors via a nonparametric bootstrap procedure in such instances. \cite{Efron86}

\subsection{Characterizing effect heterogeneity}
The presumed monotonicity of $\Delta_{\gamma; h}(a)$ in $h(a)$ gives rise to a difference-in-difference parameter to quantify heterogeneity in treatment effects across levels of engagement. In particular, define:
\begin{eqnarray*}
	\xi_ {\gamma; h} &=& \Delta_{\gamma; h}(1) - \Delta_{\gamma; h}(0) = \Delta_{\gamma; h}(1)(1 - \gamma).
\end{eqnarray*}
This value (and/or the endpoints of its respective confidence interval) can be used in two fundamental ways. The first is to quantify the extent of treatment effect heterogeneity across engagement at a particular value of $\gamma$. The second is to search for the range of $\gamma$ (if one exists) at which a clinically relevant value of $\xi_{\gamma; h}$ is achieved (estimated). Of note, a test of the hypothesis $H_0 : \xi_{\gamma; h} = 0$ is equivalent to a test of $H_0 : \dI = 0$ for reasons analogous to those outlined Corollary 3. Thus, the coverage properties of Wald-based confidence intervals for $\xi_\gamma$ are asymptotically equivalent to those of Wald-based confidence intervals for $\Delta_\gamma(a)$ and $\dI$.

\noindent \section{Simulation studies}

In this section, we conduct a set of simulation studies in order to evaluate the finite-sample performance of the sensitivity analysis procedure. We vary the sample size across three different levels ($N$ = 50, 250, and 1,000), the true value of $\gamma_0 = \Delta(0)/\Delta(1)$, across five different levels, evenly spaced between zero and one (inclusive), and the instrument strength (low, moderate, and high), as given by average engagement. Treatment randomization was generated as $Z \sim \text{Bernoulli}(0.5)$. Unmeasured confounding was represented by a single standard normal covariate, $U$, from which potential engagement under randomization to $Z = 1$ was generated as a semi-continuous variable (with nonzero masses at zero and one):
\begin{eqnarray*}
\Pro(A^{z = 1} = 1|U = u) &=& \text{expit}(\alpha_{01} + \alpha_{11} u),\\
\Pro(A^{z = 1} = 0|U = u, A^{z = 1} \neq 1) &=& \text{expit}(\alpha_{00} + \alpha_{10} u),\\
\text{logit}(A^{z = 1})  &\sim& \mathcal{N}(\mu = \alpha_0 + \alpha_1 U, \sigma^2 = \sigma_A^2) \text{ if } A^{z = 1} \neq 0 \text{ and } A^{z = 1} \neq 1.
\end{eqnarray*}
We choose $h_0(\cdot)$ to be the identity function throughout all simulation scenarios without loss of generality, as mentioned in Section 2.5. The outcome variable was generated as $Y \sim \mathcal{N}(\mu = \beta_0 + \beta_1 Z + \beta_2 A + \beta_3 U + \beta_4 L, \sigma^2 = \sigma_Y^2)$, where $L$ denotes a single standard normal covariate associated with $Y$ alone. In generating the engagement variable, we fixed certain parameters as follows, setting $\alpha_{01} = \alpha_{00} = -2$, $\alpha_{10} = \alpha_{11} = 1$, $\alpha_1 = 0.8$, $\sigma_A = 0.2$. Instrument strength is controlled by variations in $\alpha_0$, which we set as $\alpha_0 = -1.05$ for low instrument strength ($\mu_A \approx 0.35$), $\alpha_0 = -0.05$ for moderate instrument strength ($\mu_A \approx 0.50$), and $\alpha_0 = 1.9$ for high instrument strength ($\mu_A \approx 0.75$).

In generating the outcome, we fixed $\beta_0 = 9$, $\beta_3 = 0.2$, $\beta_4 = 0.3$, and $\sigma_Y = 0.8$. The true (non-identifiable) value of $\gamma_0$, is governed by $(\beta_1, \beta_2)$, and is specifically given by $\gamma_0 = \beta_1/(\beta_2 + \beta_2)$. We therefore select $(\beta_1, \beta_2)$ under five cases in order to vary $\gamma_0$ between zero and one; the $i^{\text{th}}$ case utilizes $\beta_1 = (1 - i)/5$ and $\beta_2 = -(4/5 + \beta_1)$ for $i = 1, \dots, 5$. Note that the NECE, $\Delta(0)$, is given by $\beta_1$ and the ECCE, $\Delta(1)$, is given by $\beta_1 + \beta_2 = -0.8$. This simulation is designed to loosely mirror our subsequent application to REACH in Section 6.

Under each simulation scenario, we estimate the ITT effect based on linear regression, adjusting linearly for $L$; $\Delta_\gamma(a)$ is estimated using the approach described in Section 4 for a range of $\gamma$ and a range of $a$ spanning between zero and one. In all scenarios, we use $K =$ 1,000 Monte Carlo iterations, employing $B = $ 500 bootstrap replicates. We extract the average point estimates and Monte Carlo empirical standard errors, along both the average large-sample theory based standard error and the bootstrap standard error (for comparison to each other and to the empirical standard error).

\begin{table}
\caption{Simulation study results under correct selection of $\gamma$. Depicted are the average point estimates, empirical (Monte Carlo) standard error, as well as the average large-sample-theory (LST) and bootstrap (B) standard errors for both $\Delta_\gamma(0)$ and $\Delta_\gamma(1)$ under different scenarios. Note that in all cases, the ECCE is given by $\Delta_\gamma(1) = -0.80$.}
{\footnotesize
$$
\begin{tabular}{cccrccccccccc}
~ & ~ & ~ & ~ & \multicolumn{4}{c}{$\widehat{\Delta}_\gamma(0)$} & & \multicolumn{4}{c}{$\widehat{\Delta}_\gamma(1)$} \\ \hline

$\gamma$ & NECE & $\mu_A$ & $N$ & Est. & ESE & $\widehat{\text{SE}}_{LST}$ & $\widehat{\text{SE}}_B$ & &   Est. & ESE & $\widehat{\text{SE}}_{LST}$ & $\widehat{\text{SE}}_B$ \\ \hline
0.00  & 0.00 & 0.35 & 50 & 0.000  &  0.000  &  0.000  & 0.000 & &  -0.832  &  0.773  & 0.722 &  0.755 \\
0.00  & 0.00 & 0.35 & 250 & 0.000  &  0.000  &  0.000 & 0.000  & &	-0.801  &  0.285  & 0.300 &  0.302 \\ 
0.00  & 0.00 & 0.35 & 1,000 & 0.000  &  0.000  &  0.000  & 0.000 & & -0.800  &  0.151  & 0.151 &  0.150 \\

0.00  & 0.00 & 0.50 & 50 & 0.000  &  0.000  &  0.000  & 0.000 & & -0.816  &  0.497  & 0.496 &  0.503 \\
0.00  & 0.00 & 0.50 & 250 & 0.000  &  0.000  &  0.000  & 0.000 & & -0.800  &  0.201  & 0.215 &  0.211  \\ 
0.00  & 0.00 & 0.50 & 1,000 & 0.000  &  0.000  &  0.000  & 0.000 & & -0.800  &  0.105  & 0.107 &  0.105  \\

0.00  & 0.00 & 0.75 & 50 & 0.000  &  0.000  &  0.000  & 0.000 & & -0.807  &  0.348  & 0.329 &  0.326 \\
0.00  & 0.00 & 0.75 & 250 & 0.000  &  0.000  &  0.000  & 0.000 & & -0.800  &  0.133  & 0.138 &  0.140 \\ 
0.00  & 0.00 & 0.75 & 1,000 & 0.000  &  0.000  &  0.000  & 0.000 & & -0.800  &  0.069  & 0.071 &  0.070  \\

0.25  & -0.20 & 0.35 & 50 & -0.203  &  0.127  & 0.118 &  0.118  & &  -0.812  &  0.508  & 0.470 &  0.473 \\
0.25  & -0.20 &  0.35 & 250 & -0.200  &  0.049  & 0.051 &  0.051  & &  -0.800  &  0.194  & 0.206 &  0.204 \\ 
0.25  & -0.20 &  0.35 & 1,000 & -0.200  &  0.025  & 0.026 &  0.025  & &  -0.800  &  0.101  & 0.102 &  0.102 \\

0.25  & -0.20 &  0.50 & 50 & -0.202  &  0.104  & 0.096 &  0.097  & &  -0.808  &  0.417  & 0.384 &  0.386 \\
0.25  & -0.20 &  0.50 & 250 & -0.200  &  0.040  & 0.042 &  0.042  & &  -0.800  &  0.160  & 0.169 &  0.168 \\ 
0.25  & -0.20 &  0.50 & 1,000 & -0.200  &  0.021  & 0.021 &  0.021  & &  -0.800  &  0.084  & 0.084 &  0.084 \\

0.25  & -0.20 &  0.75 & 50 & -0.201  &  0.080  & 0.074 &  0.074  & &  -0.805  &  0.319  & 0.296 &  0.295 \\
0.25  & -0.20 &  0.75 & 250 & -0.200  &  0.031  & 0.033 &  0.032  & &  -0.799  &  0.123  & 0.131 &  0.129 \\ 
0.25  & -0.20 &  0.75 & 1,000 & -0.200  &  0.016  & 0.016 &  0.016  & &  -0.800  &  0.064  & 0.065 &  0.064 \\

0.50  & -0.40 & 0.35 & 50 & -0.403  &  0.191  & 0.175 &  0.176  & &  -0.806  &  0.382  & 0.350 &  0.352 \\
0.50  & -0.40 & 0.35 & 250 & -0.400  &  0.074  & 0.078 &  0.077  & &  -0.799  &  0.148  & 0.156 &  0.155 \\ 
0.50  & -0.40 & 0.35 & 1,000 & -0.400  &  0.039  & 0.039 &  0.039  & &  -0.800  &  0.077  & 0.077 &  0.077 \\

0.50  & -0.40 & 0.50 & 50 & -0.402  &  0.172  & 0.158 &  0.159  & &  -0.805  &  0.344  & 0.316 &  0.317 \\
0.50  & -0.40 & 0.50 & 250 & -0.400  &  0.067  & 0.070 &  0.070  & &  -0.799  &  0.133  & 0.140 &  0.140 \\ 
0.50  & -0.40 & 0.50 & 1,000 & -0.400  &  0.035  & 0.035 &  0.035  & &  -0.800  &  0.070  & 0.070 &  0.070 \\

0.50  & -0.40 & 0.75 & 50 & -0.402  &  0.147  & 0.136 &  0.136  & &  -0.803  &  0.295  & 0.271 &  0.272 \\
0.50  & -0.40 & 0.75 & 250 & -0.400  &  0.057  & 0.060 &  0.060  & &  -0.799  &  0.114  & 0.120 &  0.119 \\ 
0.50  & -0.40 & 0.75 & 1,000 & -0.400  &  0.030  & 0.030 &  0.030  & &  -0.800  &  0.060  & 0.060 &  0.060 \\

0.75  & -0.60 & 0.35 & 50 & -0.602  &  0.230 & 0.211 &  0.212  & &  -0.803  &  0.307  & 0.281 &  0.282 \\
0.75  & -0.60 & 0.35 & 250 & -0.599  &  0.089  & 0.093 &  0.093  & &  -0.799  &  0.119  & 0.125 &  0.125 \\ 
0.75  & -0.60 & 0.35 & 1,000 & -0.600  &  0.047  & 0.047 &  0.047  & &  -0.800  &  0.062  & 0.062 &  0.062 \\

0.75  & -0.60 & 0.50 & 50 & -0.602  &  0.220  & 0.202 &  0.203  & &  -0.803  &  0.294  & 0.269 &  0.271 \\
0.75  & -0.60 & 0.50 & 250 & -0.599  &  0.086  & 0.090 &  0.090  & &  -0.799  &  0.114  & 0.119 &  0.119 \\ 
0.75  & -0.60 & 0.50 & 1,000 & -0.600  &  0.045  & 0.045 &  0.045  & &  -0.800  &  0.060  & 0.060 &  0.060 \\

0.75  & -0.60 & 0.75 & 50 & -0.602  &  0.206  & 0.189 &  0.189  & &  -0.802  &  0.274  & 0.254 &  0.253 \\
0.75  & -0.60 & 0.75 & 250 & -0.599  &  0.080  & 0.084 &  0.084  & &  -0.799  &  0.106  & 0.111 &  0.111 \\ 
0.75  & -0.60 & 0.75 & 1,000 & -0.600  &  0.042  & 0.042 &  0.042  & &  -0.800  &  0.056  & 0.056 &  0.056  \\

1.00  & -0.80 & 0.35 & 50 & -0.801  &  0.257  & 0.236 &  0.236  & &  -0.801  &  0.257  & 0.236 &  0.236 \\
1.00  & -0.80 & 0.35 & 250 & -0.799  &  0.100  & 0.104 &  0.104  & &  -0.799  &  0.100  & 0.104 &  0.104 \\ 
1.00  & -0.80 & 0.35 & 1,000 & -0.800  &  0.052  & 0.052 &  0.052  & &  -0.800  &  0.052  & 0.052 &  0.052 \\

1.00  & -0.80 & 0.50 & 50 & -0.801  &  0.257  & 0.236 &  0.236  & &  -0.801  &  0.257  & 0.236 &  0.236 \\
1.00  & -0.80 & 0.50 & 250 & -0.799  &  0.100  & 0.104 &  0.104  & &  -0.799  &  0.100  & 0.104 &  0.104 \\ 
1.00  & -0.80 & 0.50 & 1,000 & -0.800  &  0.052  & 0.052 &  0.052  & &  -0.800  &  0.052  & 0.052 &  0.052 \\

1.00  & -0.80 & 0.75 & 50 & -0.801  &  0.257  & 0.236 &  0.236  & &  -0.801  &  0.257  & 0.236 &  0.236 \\
1.00  & -0.80 & 0.75 & 250 & -0.799  &  0.100  & 0.104 &  0.104  & &  -0.799  &  0.100  & 0.104 &  0.104 \\ 
1.00  & -0.80 & 0.75 & 1,000 & -0.800  &  0.052  & 0.052 &  0.052  & &  -0.800  &  0.052  & 0.052 &  0.052 \\ \hline
\end{tabular}
$$
}
\end{table}

Key results of this study are summarized in Table 2. When the value of $\gamma$ is correctly specified, the sensitivity analysis approach is able to correctly capture the true NECE and ECCE, with standard errors reflecting the true repeat-sample variability (as represented by the empirical Monte Carlo standard error). The level of bias associated with smaller sample sizes is consistent with prior insights regarding bias of the Wald estimator.\cite{Buse92} The standard errors based on both large-sample theory and the nonparametric bootstrap performed comparably well with neither being obviously superior to the other in any case.

Figures 4 and 5 demonstrate the properties of our estimation procedure under various specified levels of $\gamma$ for the cases of lower engagement and higher engagement (respectively). Of note, specifying $\gamma = \gamma_0$ results in nearly unbiased estimation, as expected (and as demonstrated in Table 2). Further, selections of $\gamma$ that are closer to the correct value ($\gamma_0$) result in less bias than selections of $\gamma$ that are further from $\gamma_0$. The curvilinear relationship between $\gamma$ and $\Delta_\gamma(a)$ for fixed values of $a$ is reflected in these figures, and is more pronounced for the lower engagement scenarios. Moreover, the figures also reflect the stationary property discussed in Section 3, whereby $\Delta_\gamma(\mu_A) = \dI$.

\begin{figure}[h!]
\centering
\includegraphics[width = 6.5in]{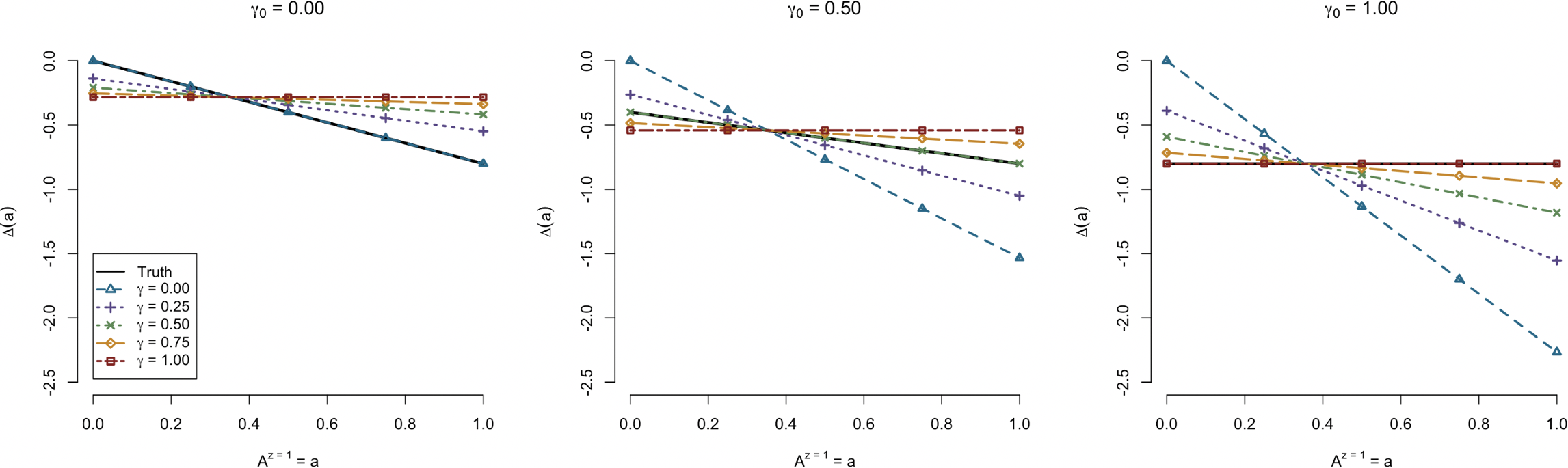}
\caption{Results of the simulation study (lower levels of average engagement) for $\gamma_0 = 0$ (left), $\gamma_0 = 0.50$ (center), and $\gamma_0 = 1$ (right). Plotted are the average point estimates across different levels of $a$ for various selections of $\gamma$.}
\end{figure}

\begin{figure}[h!]
\centering
\includegraphics[width = 6.5in]{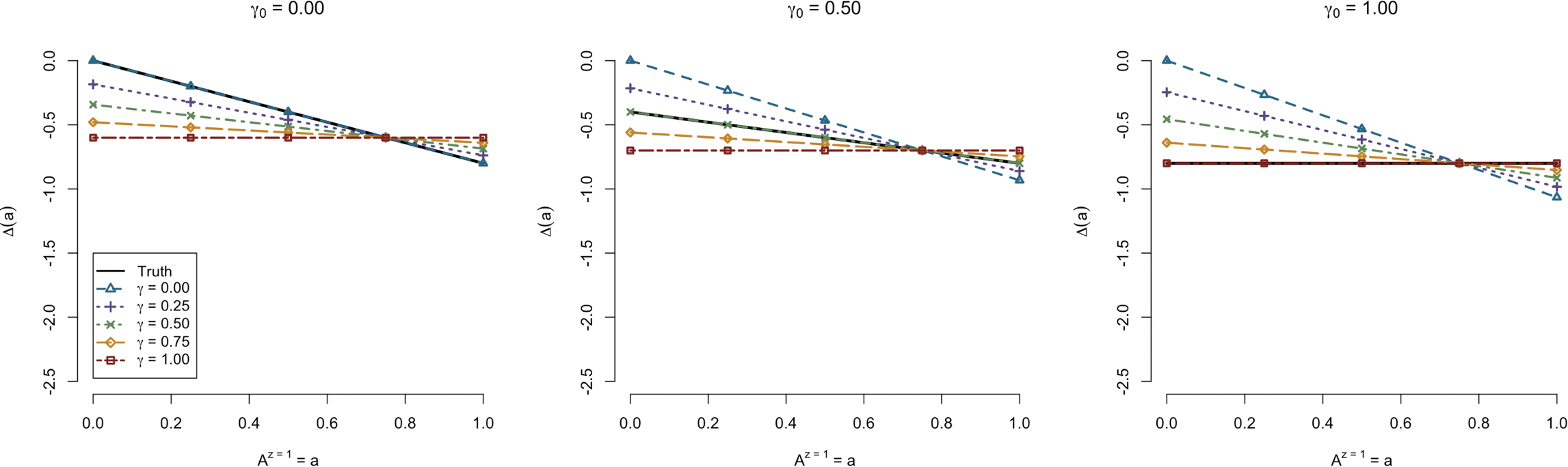}
\caption{Results of the simulation study (higher levels of average engagement) for $\gamma_0 = 0$ (left), $\gamma_0 = 0.50$ (center), and $\gamma_0 = 1$ (right). Plotted are the average point estimates across different levels of $a$ for various selections of $\gamma$.}
\end{figure}

\noindent \section{Application to REACH}

\subsection{Description of data and methods}

We illustrate our approach by applying it to the REACH study. For the purposes of analyses appearing in this paper, the intervention, $Z$, is characterized by randomization to either a control condition ($N_0 = 106$), or to the REACH intervention ($N_1 = 109$). Subjects in the intervention arm received daily text messages over a period of six months including one-way messages that provided self-care information and encouragement and two-way messages that asked about medication adherence. At the end of each week, subjects in the intervention arm received adherence feedback based on his or her responses that week. Subjects considered for this analysis all had uncontrolled HbA1c at baseline, characterized as either meeting or exceeding 8.5\%. The outcome, $Y$, is given by HbA1c six-months post randomization; we consider an average causal effect of REACH on HbA1c of 0.50\% to be clinically meaningful. Subject-specific engagement, $A$, was measured as the proportion of text messages responded to (applicable only to subjects assigned to the intervention arm). The engagement values for subjects withdrawing prior to the six-month period were considered pragmatically; such subjects were considered as having zero-engagement for the remainder of the six-month period post-withdrawal.

Missing data were addressed using multiple imputation with chained equations based on baseline demographic, socioeconomic, and clinical characteristics; the imputation procedure was aggregated with the bootstrap using the pooled-sample nested approach recommended by Schomaker et al.\cite{Schomaker18} We report 95\% confidence intervals based on the 0.025 and 0.975 quantiles of $B$ = 500 nonparametric bootstrap replicates and $M = 500$ multiple-imputation iterations. We estimate the ITT using linear regression, adjusting for baseline HbA1c using a natural cubic spline with knots at the three inner quartiles (8.90\%, 9.70\%, and 11.1\%), as recommended by Harrell.\cite{Harrell01} We estimate $\Delta_\gamma(a)$ for $a = 0, 0.5, 0.814,$ and $1$ (0.814 denotes the average response rate), varying $\gamma$ between zero (classic IV analysis) and one (ITT analysis), and modeling $\Delta_\gamma(a)$ linearly in $a$---i.e., choosing $h(\cdot)$ to be the identity function. All analyses were performed using R, version 4.0.2.\cite{R2020} 

\subsection{Results}

Figure 6 presents a histogram of the distribution of engagement across patients in the intervention arm. The mean subject-specific text message response rate was 81.4\% (SD: 23.1\%). The median response rate was 91.5\%, with interquartiles given by 74.0\% and 91.5\%. Approximately 11\% of subjects had a response rate no higher than 50\%.

\begin{figure}[h!]
\centering
\includegraphics[width = 2.75in]{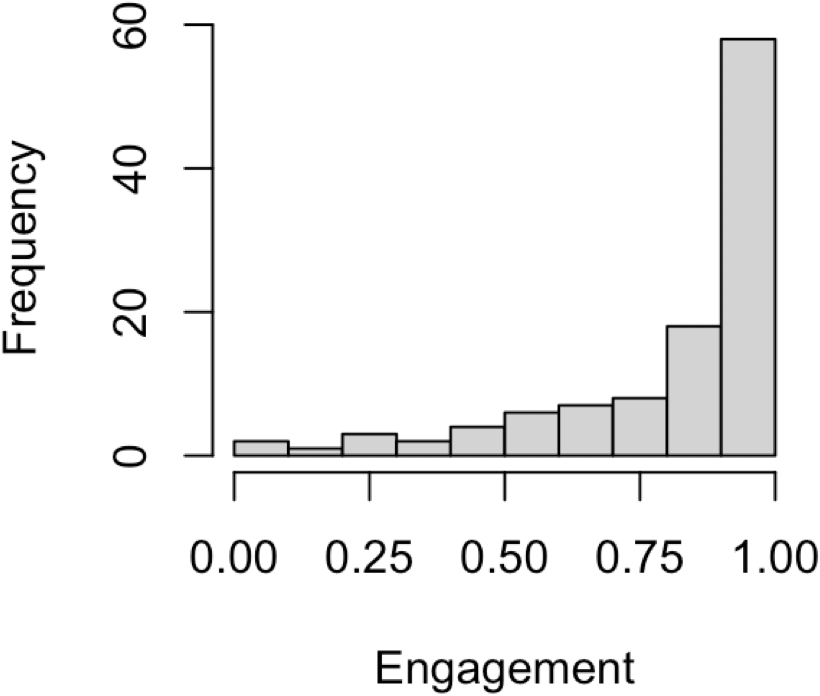}
\caption{Histogram of subject-specific engagement within the REACH intervention group, defined as proportion of text messages responded to over the six-month period.}
\end{figure}

The ITT was estimated to be -0.761\% (95\% CI: [-1.30\%, -0.24\%]; p = 0.0049). The results of the sensitivity analysis under different sensitivity parameters are shown in Table 3. First, we note that the results of our analysis confirm the stationary property described in Section 3, and are consistent with the derived bounds. In practice, one would likely not consider such a broad range of $\gamma$, but instead a narrower range thought to be more plausible for the application at hand. In this example, $\gamma$ in more moderate ranges are more plausible as compared to either of the extremes, and so the three intermediate examples (Figure 7) better illustrate the proper use of the sensitivity analysis approach in this setting than the extreme cases, which are shown more for the purpose of illustration.

\begin{table}
	\caption{Results from the REACH study. Presented are the estimated local average treatment effects and respective quantile-based bootstrap 95\% confidence intervals for different values of theoretical engagement with the treatment, across different levels of $\gamma$.}
	{\footnotesize
		$$
		\begin{tabular}{cccccccccccc}
		~ & \multicolumn{2}{c}{$\widehat{\Delta}_\gamma(0)$} & & \multicolumn{2}{c}{$\widehat{\Delta}_\gamma(0.25)$} & & \multicolumn{2}{c}{$\widehat{\Delta}_\gamma(0.814)$} & & \multicolumn{2}{c}{$\widehat{\Delta}_\gamma(1.00)$} \\ \hline
		
		$\gamma$ & Est. & 95\% CI & & Est. & 95\% CI & & Est. & 95\% CI & & Est. & 95\% CI \\ \hline
		0.00  & 0.00 & NA & & -0.47 & [-0.80, -0.14] & & -0.76 & [-1.30, -0.23] & & -0.94 & [-1.60, -0.29]  \\
		0.25  & -0.22 & [-0.38, -0.07] & & -0.55 & [-0.94, -0.17] & & -0.76 & [-1.30, -0.23] & & -0.88 & [-1.51, -0.27] \\
		0.50  & -0.42 & [-0.72, -0.13] & & -0.63 & [-1.07, -0.19] & & -0.76 & [-1.30, -0.23] & & -0.84 & [-1.43, -0.26] \\
		0.75  & -0.60 & [-1.02, -0.18] & & -0.70 & [-1.19, -0.22] & & -0.76 & [-1.30, -0.23] & & -0.80 & [-1.36, -0.25] \\
		1.00  & -0.76 & [-1.30, -0.23] & & -0.76 & [-1.30, -0.23] & & -0.76 & [-1.30, -0.23] & & -0.76 & [-1.30, -0.23] \\ \hline	
	\end{tabular}$$
	}
\end{table}

\begin{figure}[h!]
	\centering
	\includegraphics[width = 6.5in]{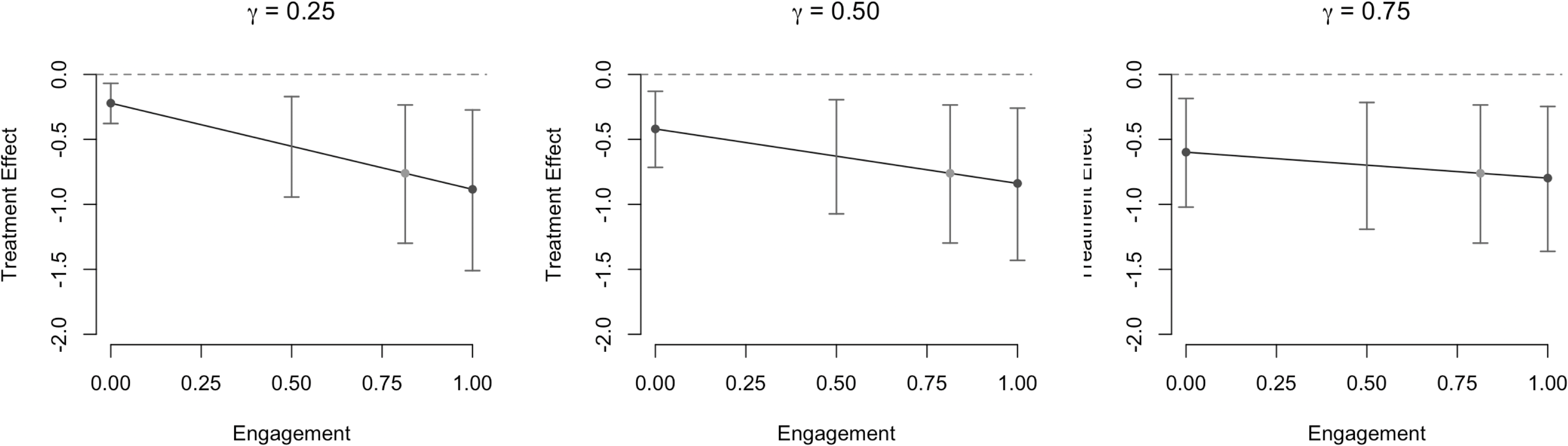}
	\caption{Local average treatment effects and corresponding confidence intervals across levels of $A^{z = 1} = a$ for three different values of the sensitivity parameter, $\gamma$. Also depicted in each plot are the estimated NECE and ECCE (dark gray points) and ITT (light gray point), and quantile-based bootstrap 95\% confidence intervals for local average treatment effects at specified levels of $a$.}
\end{figure}

We can glean a number of insights from this set of sensitivity analyses. For instance, if $\gamma = 0.50$, engagement levels meeting or exceeding 19.2\% are associated with treatment effect estimates exceeding the clinically meaningful threshold of 0.50\%. On the other hand, if $\gamma = 0.75$, all levels of engagement are associated with treatment effect estimates exceeding that threshold. One can further derive insights regarding local average treatment effects on the basis of confidence interval endpoints instead of point estimates. For instance, if $\gamma = 0.25$, engagement levels under 10.8\% rule out both a null effect and a clinically meaningful effect.

In addition to characterizing information regarding local average treatment effects at fixed levels of $\gamma$, one can also search for the values of $\gamma$ that meet a particular criteria, as described in Section 4.1. We find, for instance, that values of $\gamma$ smaller than 0.690 are associated with an estimated value of $\xi$ exceeding 0.25\%. In other words, we estimate that the average treatment effect among never-engagers needs to be smaller than 69.0\% of that among engagement-compliant in order for the treatment effect of REACH on HbA1c to differ between those two groups by 0.25\%.

\noindent \section{Discussion}

In this paper, we have derived and presented a sensitivity analysis approach to accommodate departures from the exclusion restriction when estimating local average treatment effects, with specific applications to engagement in mobile health interventions. In this setting, the principal stratification framework is simplified by the impossibility of engagement with an intervention not received. Hence, local average treatment effects can be characterized conditional on a single (partially latent, and possibly transformed) variable. Placing reasonable bounds on the sensitivity parameter in turn results in conceptually intuitive bounds on the average causal effect.

Our proposed approach is designed to aid insights regarding average causal effects of the intervention at various levels of engagement with the intervention, particularly when violations to the exclusion restriction assumption cannot be ruled out. We presented asymptotic theory that holds for invertible transformations of the post-randomization variable; our simulations and application focused on continuous, linear treatment of engagement, under which our sensitivity procedure appeared to have desirable finite-sample properties. Though misspecification of $h(\cdot)$ was not the focus of this work, further study of different transformations and their possible advantages could be of interest for future studies.

Our illustrative example demonstrates various ways that this approach can be used to glean insights regarding treatment effect heterogeneity across levels of engagement. Naturally, the tighter the bounds on $\gamma$ that can reasonably be considered in practice, the more robust and precise the conclusions that can be derived. Bounds on $\gamma$ are best dictated by the particular example to which the sensitivity analysis is being applied. For interventions featuring mostly two-way content, it stands to reason that lower values of $\gamma$ may be more reasonable as compared to interventions comprising mostly one-way content. We note that engagement with an intervention is more of an abstract concept than response to two-way text messages. For instance, a subject's true engagement in the abstract sense is partially reflected by his or particular level of attention to text messages (length of time read). Proportion of text messages receiving a response is one of many possible objective measures of engagement, but does not necessarily serve as a perfect surrogate for instrinsic engagement in the most abstract sense. The results of a sensitivity analysis are driven in large part by the average level of the engagement metric in the study. In the particular case of the REACH study, text message response rates tended to be high, such that the ECCE was less sensitive to fluctuations in the sensitivity parameter as compared to the NECE. Had the average engagement rate been lower, the ECCE would have been more sensitive to fluctuations in $\gamma$.

Specific procedures for sensitivity analyses have long been an area of interest in methodological causal inference research, many times in the context of violations to the assumption of ignorability/no unmeasured confounding.\cite{Lin98, Imai10, Jo11, Stuart15, Dorie16} Prior work has investigated the sensitivity of non-IV based causal inference approaches when the exclusion restriction is not satisfied. \cite{Millimet13} Many of the sensitivity analysis procedures developed for departures from the IV assumptions are applicable only to the case of four discrete principal strata (e.g, in the example of treatment compliance). In this case, Angrist et al. are able to express the bias of the IV estimand explicitly in terms of the direct effect of the instrument on the outcome and the odds of being a non-complier.\cite{Angrist96} Again in the case of four discrete principle strata, Baiocchi et al. propose a number of sensitivity analysis procedures for departures to IV assumptions, and Stuart et al. demonstrate how the exclusion restriction can be replaced with an alternative assumption referred to \textit{principal ignorability} when predictors of stratum membership are thought to be well understood.\cite{Baiocchi14, Stuart15} The corresponding approach is analogous to propensity score methods in order to predict subgroup-specific stratum membership. Such approaches are best suited for settings in which predictors of principal stratum are known and measured, and not when substantive unmeasured confounding is suspected as in the case of continuous measures of engagement with an intervention (e.g., REACH). Generalizing the methodology of Stuart et al. to the setting in which strata are defined by a single continuous post-randomization variable, however, may serve as a potential topic of interest for future research.\\
\end{spacing}

\clearpage

\noindent \textbf{Acknowledgements}

\setlength{\parindent}{0.7cm} 

{
	
	\singlespacing
	
	\noindent This research was funded by the National Institutes of Health NIH/NIDDK R01DK100694 and NIH/NIDDK Center for Diabetes Translation Research Pilot and Feasibility Award P30DK092986. Dr. Lyndsay Nelson was supported by a career development award from NIH/NHBLI (K12HL137943). The content is solely the responsibility of the authors and does not necessarily represent the official views of the National Institutes of Health.
	
}

\clearpage

\bibliographystyle{plain}

{\onehalfspacing

\noindent \bibliography{bib}

}

 \end{document}